\documentclass[reqno,12pt]{amsart}

\usepackage{showlabels}

\newtheorem{theorem}{Theorem}[section]
\newtheorem{proposition}{Proposition}[section]

\newcommand{\RR}{\mathbb{R}}

\newcommand{\dell}{\partial}

\newcommand{\tr}{\mathrm{tr}\,}
\newcommand{\vL}{\Delta_{\mathbb L}}

\newcommand{\calC}{{\mathcal C}}
\newcommand{\calD}{{\mathcal D}}

\newcommand{\calH}{{\mathcal H}}

\newcommand{\calN}{{\mathcal N}}

\newcommand{\calS}{{\mathcal S}}

\newcommand{\calY}{{\mathcal Y}}

\begin{document}

\title[Non-CMC Solutions on Asymptotically Euclidean Manifolds]{Non-CMC Solutions of the Einstein
Constraint Equations on Asymptotically Euclidean Manifolds}

\author[J. Dilts]{James Dilts}
\address{University of Oregon}
\email{jdilts@uoregon.edu}

\author[J. Isenberg]{Jim Isenberg}
\address{University of Oregon}
\email{jisenberg@uoregon.edu}

\author[R. Mazzeo]{Rafe Mazzeo}
\address{Stanford University}
\email{mazzeo@math.stanford.edu}

\author[C. Meier]{Caleb Meier}
\address{UC San Diego}
\email{c1meier@math.ucsd.edu}

%\thanks{}

%\date{\today}


\maketitle

\vspace*{-1.2cm}

\begin{abstract}
In this note we prove two existence theorems for the Einstein constraint equations on asymptotically
Euclidean manifolds. The first is for arbitrary mean curvature functions with restrictions on the size of the
transverse-traceless data and the non-gravitational field data, while the second assumes a near-CMC condition,
with no other restrictions.
\end{abstract}

%%%%%%%%%%%%%%%%%%%%%%%%%%%%%%%%%%%%%%%%%%%%%%%%%%%%%%%%%%%%%%%%%%%%%%%%%%%%%%
\section{Introduction}
\label{Intro}
%%%%%%%%%%%%%%%%%%%%%%%%%%%%%%%%%%%%%%%%%%%%%%%%%%%%%%%%%%%%%%%%%%%%%%%%%%%%%%
This paper capitalizes on several recent advances concerning the existence of solutions of the Einstein
constraint equations using the conformal method. Using these new techniques we construct solutions
of these equations, either in the vacuum setting or else with coupled non-gravitational  fields, on asymptotically Euclidean
manifolds under two separate sets of hypotheses: either the mean curvature function $\tau$ is arbitrary
but the transverse-traceless part of the data $\sigma$ (and the non-gravitational  field densities) are very small
(``far-CMC" data) or else smallness assumptions are placed on $d\tau/\tau$ (the ``near-CMC" case).  We do this
by adapting the methods of Holst-Nagy-Tsogterel \cite{closed-fcmc} and Maxwell \cite{vac-fcmc}.

We recall that a Riemannian manifold $(M,\hat g)$ and a symmetric $2$-tensor $\hat K$ on $M$ satisfy the Einstein
constraint equations  with non-gravitational energy-density $\hat \rho$ (a scalar function) and non-gravitational
momentum density $\hat J$ (a vector field) if
\begin{equation}
\begin{aligned}
|\hat K|_{\hat g}^2 - (\tr_{\hat g} \hat K)^2 & = R(\hat g) - \hat  \rho, \\
\mathrm{div}_{\hat g}\, \hat K - \nabla \tr_{\hat g} \hat K & = \hat J.
\end{aligned}
\label{vce}
\end{equation}
Note that we work here only with non-gravitational fields which require no further constraints; this is the case for fluid fields, for example. (These results easily extend to theories such as Einstein-Maxwell which do introduce extra constraints; for brevity, we do not treat such cases here.) The cosmological constant $\Lambda$ is assumed to vanish because our interest here is exclusively
with asymptotically Euclidean data.

One uses the conformal method to generate an initial data set  $(M, \hat g, \hat K, \hat \rho, \hat J)$ which satisfies the constraints \eqref{vce} by first (freely) choosing the \emph {conformal data}, which includes a Riemannian manifold $(M,g)$, a symmetric tensor $\sigma$ which is transverse ($\mathrm{div}_g \sigma =0$) and traceless ($\tr_g \sigma =0$) with respect to $g$, a scalar function $\tau$ (the mean curvature), a non-negative scalar function $\rho$, and a vector field $J$. One then seeks solutions $\phi$ (a positive scalar) and $W$ (a vector field) of the \emph{conformal constraint equations}
\begin{equation}
\begin{aligned}
(i) \ & \Delta_{g} u - c_n R_g u + c_n |\sigma + \calD W|^2_g u^{-N-1} - b_n \tau^2 u^{N-1} + c_n \rho u^{-\frac{N}{2}}= 0, \\
(ii) \ & \Delta_{\mathbb L} W + \frac{n-1}{n} u^{N} \nabla \tau + J = 0.
\end{aligned}
\label{ce2}
\end{equation}
Here $R_g$ is the scalar curvature of $g$, $\calD$ is the conformal Killing operator acting on vector fields
\[
(\calD W)_{ij} := \nabla_i W_j + \nabla_j W_i - \frac{2}{n}(\mathrm{div}_{g} W ) \, g_{ij},
\]
$ \Delta_g$ is the scalar Laplacian, $\vL := -\mathrm{div} \circ \calD$ is the vector Laplacian,
and the constants $N$,  $c_n$, and $b_n$  are given by
\[
N := \frac{2n}{n-2}, \quad c_n := \frac{n-2}{4(n-1)}, \quad b_n:= \frac{n-2}{4n}.
\]
If $(\phi,W)$ is a solution of \eqref{ce2}, then the initial data set
\[
\hat g = u^{\frac{4}{n-2}} g, \quad \hat K = u^{-2\frac{n+2}{n-2}}(\sigma + \calD W)^{ij}+\frac{\tau}{n} u^{\frac{4}{n-2}} g^{ij},\quad  \hat \rho = u^{-\frac{3}{2}N+1} \rho, \quad \hat J =u^{-N}J,
\]
satisfies the Einstein constraints \eqref{vce}.

For convenience below, if $v$ is any positive function (in a suitable function space), then we let $W(v)$ denote the
solution of equation (ii) of \eqref{ce2}, where this function $v$ is inserted on the right hand side.  Similarly, we write the Lichnerowicz
operator on the left of equation (i) of \eqref{ce2} as $\calN(u,W)$.  Thus a solution $(u,W)$ of the coupled
system \eqref{ce2} corresponds to a solution $u$ of the single nonlocal equation $\calN(u, W(u)) = 0$.

When $\tau$ is constant (the CMC case), these equations decouple, and it is possible to obtain incisive results for this case; see \cite{const-compact,const-asympt-flat}.
Similarly, when $\nabla \tau/\tau$ is suitably small (this is known as the near-CMC case), then many further results have been
obtained using  perturbation methods. The recent advances, stemming from the papers of Holst, Nagy and Tsogterel \cite{closed-fcmc},
later refined and simplified by Maxwell \cite{vac-fcmc}, treat the case in which $\tau$ is allowed to vary with no restrictions; these
results are still perturbative in a different sense because they require $\sigma, \rho$, and $J$  to be very small (this is a special
case of the general far-CMC case). These arguments rely on the Yamabe positivity of the underlying conformal class $[g]$
and on  these various smallness conditions to construct barriers. This particular far-CMC scenario has now been worked out in
several settings. The original papers treat the case where $M$ is closed; the more recent papers of Holst, Meier and Tsogtgerel
\cite{compact-bound2} and Dilts \cite{compact-bound1} treat the case where $M$ is a manifold with boundary, considering a wide range of boundary conditions;
finally, Leach \cite{asympt-cyl} has dealt with the case where $M$ is complete with cylindrical ends. In this paper we continue this line of research and
prove an existence result in this far-CMC case for manifolds with asymptotically Euclidean ends, which is one of the standard
and most important settings in relativity.  The new issue to be faced here is the way that barriers must be constructed near infinity.
This is similar to what must be done in the cylindrical case, but the argument here is simpler than in \cite{asympt-cyl}.
We also determine the precise asymptotics of solutions. A near-CMC result is proved here using very similar methods
(we recall that \cite{const-asympt-flat} contains other near-CMC results for asymptotically Euclidean data sets.)

We now state our main results.  The precise definitions of asymptotically Euclidean metrics, and of the weighted Sobolev spaces
appearing in these statements, are all given in the next section.
\begin{theorem}{\bf (Far-from-CMC)}
\label{FarC}
Suppose that $(M^n,g)$ is a $W^{2,p}_{\gamma}$ asymptotically Euclidean (AE) metric with positive Yamabe invariant,
where $p>n$ and $\gamma \in (2-n,0)$, and set $\delta = \gamma/2$. Fix data $\tau \in W^{1,p}_{\delta-1}$, $\sigma \in
L_{\delta-1}^{\infty}$, nonnegative $\rho \in L^{\infty}_{2\delta-2}$ and $ J \in L^{p}_{\delta-2}$, and assume that
$\|\sigma\|_{L^{\infty}_{\delta-1}}$, $\|\rho\|_{L^{\infty}_{\delta-2}}$ and $\|J \|_{L^{p}_{\delta-2}}$ are sufficiently small (depending
on $\tau$, $g$ and $n$). Then there exists a solution $(\phi,W)$  to \eqref{ce2} with $W \in W^{2,p}_{\delta}$, $\phi > 0$
and $\phi- A_j \in W^{2,p}_{\delta}$ for some constant $A_j > 0$ on each end $E_j$ of $M$.
\label{ffcmc}
\end{theorem}

\begin{theorem}{(\bf Near-CMC)}
\label{NearC}
Let $(M,g)$ be AE as in Theorem \ref{ffcmc}, with $\gamma$, $\delta$ as above.
Assume too that $\tau \in W^{1,p}_{\delta-1}$ with $\tau - Br ^{2\delta-2}\|d\tau\|_{L^p_{\delta-2}}>0$ for some $B > 0$,
where $r$ is an everywhere positive function which is the radial distance on each end of $M$, and that
$ \sigma \in  L_{\delta -1}^{\infty}$,  $\rho \in L^{p}_{\delta-2}$ with $\rho \geq 0$ and $J \in L^{p}_{\delta-2}$.
Then there exists a solution $(\phi,W)$ to \eqref{ce2} with $W \in W^{2,p}_{\delta}$, $\phi > 0$ and
$\phi- A_j \in W^{2,p}_{\delta}$ for some constant $A_j > 0$ on each end $E_j$ of $M$.
\end{theorem}

The primary task in proving these theorems is to establish the existence of upper and lower barriers for equations \eqref{ce2}.
After discussing asymptotically Euclidean manifolds and function spaces in Section \ref{AsEuc}, and then reviewing the
mapping properties of the scalar and vector Laplacian operators on AE manifolds in Section \ref{Map}, we derive these barriers
in Section \ref{sec:barriers}. A standard fixed point theorem is then used  in Section \ref{sec:FixedPoint}
to prove Theorems \ref{FarC} and \ref{NearC}.

%%%%%%%%%%%%%%%%%%%%%%%%%%%%%%%%%%%%%%%
\section{Asymptotically Euclidean Manifolds}
\label{AsEuc}
%%%%%%%%%%%%%%%%%%%%%%%%%%%%%%%%%%%%%%%
Let $(M^n,g)$ be an asymptotically Euclidean (AE) manifold. This means that $M$ is a complete
manifold such that for some compact set $K \subset M$, the complement $M \setminus K$
has finitely many components, $E_1, \ldots, E_\ell$, where each $E_j$ is diffeomorphic
to the exterior of a ball in a Euclidean space, $E_j \cong \RR^n \setminus B_R(0)$, and on each
of these ends, the metric $g$ is asymptotic to the Euclidean metric. More precisely, recall
that a function $u \in W^{k,p}_\delta(\RR^n)$ if
\[
\sum_{|\beta| \leq k} ||r^{-\delta - \frac{n}{p} + |\beta|} \dell^\beta u||_{L^p}  < \infty.
\]
Here $r$ is a smooth positive function on $M$ which agrees with the radial function $|x|$ on
each end. To extend these spaces and norms to tensors, as needed in the characterization of the decay of the metric
$g$ above, we require this regularity and decay for each component with respect to a constant frame in the
background Euclidean metric. Thus we say that $g$ is AE of class $W^{k,p}_{\gamma}$, for some $\gamma < 0$, if
in a fixed Euclidean coordinate system for that end,
\[
\left. g \right|_{E_j}  - g_{\mathrm{Euc}} \in W^{k,p}_\gamma.
\]
The regularity of the tensor field $\hat K$ and the scalar and vector fields $\hat \rho$ and $\hat J$ are defined analogously.
We refer the reader to \cite{mass-asymp} for a survey of the well-known properties of these spaces.

We single out one fact which we used repeatedly:  if $p > n$, and $w \in W^{1,p}_\delta$ for any $\delta \in \RR$, then
\begin{equation}
|w| \leq r^{\delta} ||w||_{W^{1,p}_\delta}.
\label{weightest}
\end{equation}

The initial data set $(M,\hat g,\hat K, \Lambda, \hat \rho, \hat J)$ is said to be asymptotically  Euclidean if
$\hat g - g_{\mathrm{Euc}} \in W^{k,p}_{\gamma}$, $\hat K \in W^{k-1,p}_{\gamma-1}$, and $\rho, J \in W^{k-2,p}_{\gamma-2}$
for some $\gamma < 0$.

In the following, we always assume that $(M,g)$ is AE of class $W^{2,p}_\gamma$ with $p > n$ and $2(2-n) < \gamma < 0$,
but omit writing this explicitly. We also always assume that $\delta = \gamma/2$, so $2-n < \delta < 0$.
All results below have obvious modifications if we assume that $g$ is AE of class $W^{k,p}_\gamma$ with $k >1+ n/p$.

%%%%%%%%%%%%%%%%%%%%%%%%%%%%%%%%%%%%%%%%%%%%%%%%%%%%%%%%%%%%%%%%%%%%%%%%%%%%%%
\section{Mapping properties of the scalar and vector Laplacians}
\label{Map}
%%%%%%%%%%%%%%%%%%%%%%%%%%%%%%%%%%%%%%%%%%%%%%%%%%%%%%%%%%%%%%%%%%%%%%%%%%%%%%
The mapping properties of elliptic operators on asymptotically Euclidean spaces is now classical, going back
at least to \cite{lap-weighted}, but see also \cite{app-hor} and the appendix in \cite{CM1}. We record a few such results
needed below, pertaining to the solvability of the inhomogeneous linear equation
\[
P u = f,
\]
where $P$ is either the conformal Laplacian $\Delta_g - c_n R$ or else the vector Laplacian $\vL$.
\begin{proposition}\label{veclapfred}
If $(M,g)$ is AE, then
\begin{equation}
P: W^{2,p}_{\delta} \longrightarrow L^p_{\delta-2}
\label{fredvL}
\end{equation}
is Fredholm of index zero, and there is an a priori estimate: there is a constant $C > 0$ such that
\[
\|\psi \|_{\psi^{2,p}_{\delta}} \le C\left(\|P \psi\|_{L^p_{\delta-2}}+\|\psi\|_{L^\infty}\right)
\]
for all $\psi \in W^{2,p}_\delta$.
The map \eqref{fredvL} is an isomorphism if and only if $P$ has no nullspace in $W^{2,p}_\delta$. For $P = \Delta - c_n R$,
this is the case provided the Yamabe invariant $\mathcal Y([g]) $ is positive; while for $P = \vL$, this holds if
$(M,g)$ admits no global conformal Killing fields.  Under this isomorphism condition, the a priori estimate above
can be strengthened to
\begin{align}\label{injbound}
\|\psi \|_{W^{2,p}_{\delta}} \le C\| P \psi\|_{L^p_{\delta-2}}.
\end{align}
\end{proposition}

We record two useful corollaries.
\begin{proposition}\label{samas}
If $P$ is the conformal Laplacian $\Delta - c_nR$ and $R \geq 0$, and if $f = r^{\gamma-2} + \hat{f}$, where
$\hat{f} \in L^p_{\gamma'-2}$ for $\gamma' < \gamma$, then there is a unique solution $w$ to
$Pw = f$ with $w = c_\gamma\, r^\gamma + \hat{w}$, $c_\gamma = (\gamma^2 + (n-2)\gamma)^{-1}$, and
$\hat{w} \in W^{2,p}_{\gamma''}$ where $\gamma'' = \max\{\gamma', 2\gamma\}$ if this number is greater than $2-n$
(or else $\gamma'' \in (2-n, \gamma)$).
\end{proposition}
\begin{proof}
Write $w = c_\gamma\, r^\gamma + \hat{w}$ and let $\bar{g}$ be a $W^{2,p}$ metric which agrees with $g$ away from the ends but
is exactly Euclidean on each $E_j$. Then we must solve
\[
(\Delta - c_n R)\hat{w} = \hat{f} - c_\gamma (\Delta_{\bar{g}} - R_{\bar{g}}) r^{\gamma} - c_\gamma
( (\Delta - c_n R) - (\Delta_{\bar{g}} - c_n R_{\bar{g}})) r^\gamma.
\]
The second term on the right is $L^p$ with compact support, while the third term lies in $L^p_{2\gamma - 2}$.
Using the nonnegativity of $R$ (to rule out the kernel), the result follows from Proposition \ref{veclapfred}.
\end{proof}
\begin{proposition}\label{boundDW}
If $(M,g)$ is AE and has no conformal Killing fields, and if $f \in L^p_{\delta-2}$, then the unique solution $W\in W^{2,p}_\delta$
to $\vL W = f$ satisfies
\begin{equation}
\label{DWest}
\|\calD W\|_{\infty} \le C r^{\delta-1} ||f||_{L^p_{\delta-2}}.
\end{equation}
\end{proposition}
\begin{proof}
Combining \eqref{injbound} and \eqref{weightest}, we get
\[
r^{1-\delta}|\calD W| \leq \|\calD W\|_{L^\infty_{\delta-1}} \le C\|\calD W\|_{W^{1,p}_{\delta-1}}\le C \|W\|_{W^{2,p}_{\delta}}
\leq C ||f||_{L^p_{\delta-2}},
\]
and this gives \eqref{DWest}.
\end{proof}

%%%%%%%%%%%%%%%%%%%%%%%%%%%%%%%%%%%%%%%%%%%%%%%%%%%%%%%%%%%%%%%%%%%%%%%%%%%%%%
\section{Barriers}
\label{sec:barriers}
%%%%%%%%%%%%%%%%%%%%%%%%%%%%%%%%%%%%%%%%%%%%%%%%%%%%%%%%%%%%%%%%%%%%%%%%%%%%%%
We begin by recalling the notion of \emph{global sub- and supersolutions}. The function $\phi_+$ is called a global
supersolution for \eqref{ce2} if $\calN(\phi_+, W(\phi)) \leq 0$ whenever $0 < \phi \leq \phi_+$.
Similarly, $\phi_-$ is called a global subsolution if $\calN(\phi_-, W(\phi)) \geq 0$ whenever $\phi \in L^p$
and $\phi_- \leq \phi$.

\begin{theorem}{\bf (Far-from-CMC Global Supersolution )} \label{farcmcbarrier}
Let $(M,g)$ be AE with positive Yamabe invariant;  i.e., $\mathcal Y([g]) > 0$. If $||\sigma||_{L^\infty_{\delta-1}}$, $||J||_{L^p_{\delta-2}}$
and $||\rho||_{L^\infty_{2\delta-2}}$ are sufficiently small, then there exists a global supersolution $\phi_+>0$ with
$\phi_+ - \eta \in W^{2,p}_{\gamma}$ for some constant $\eta > 0$.
\end{theorem}
\begin{proof}
Choose a smooth, positive function $F$ which equals $r^{\gamma -2}$ outside a compact set (recall that $\gamma$ indexes
the asymptotic behavior of the AE metric). By Proposition \ref{veclapfred}, there exists a (unique) $\Psi = c_\gamma r^\gamma +
\hat{\Psi}$, with $\hat{\Psi} \in W^{2,p}_{2\gamma}$ such that
\begin{equation}
\label{Psi}
(\Delta - c_nR )\Psi = - F + c_nR,
\end{equation}
or equivalently
\begin{equation}
(\Delta - c_nR) (1+\Psi) = -F.
\end{equation}
Note that, by the maximum principle, $1+\Psi > 0$.

Now set $\phi_+ = \eta(\Psi+1)$, where the constant $\eta>0$ is to be chosen below. We claim that, for appropriate $\eta$,
$\phi_+$ is a global supersolution. To verify this, we first note that from \eqref{DWest}, with $f =\frac{n-1}{n} \phi^N \nabla
\tau + J$, we have
\begin{equation}
\label{DWestim}
\|\calD W\|_{\infty} \le C r^{\delta-1}\left(\|d \tau\|_{L^p_{\delta-2}}\|\phi\|^{N}_{\infty}+\|J\|_{L^p_{\delta-2}} \right),
\end{equation}
and hence
\[
|\sigma + \calD W|^2 \le C r^{2\delta -2}(\|d \tau\|^2_{L^p_{\delta-2}} \|\phi\|^{2N}_{\infty} +
||\sigma||_{L^\infty_{\delta-1}}^2 + \|J\|^2_{L^p_{\delta-2}}).
\]

Since $\Psi$ decays at the precise rate $r^\gamma$ (and is strictly positive), then deleting subscripts denoting the norms
for simplicity, we calculate
\begin{multline*}
\calN(\phi_+, W(\phi)) \leq  \\
- \eta \, F + r^{2\delta-2}\left( C_1 \eta^{N-1} + C_2 \eta^{-N-1} ( ||\sigma||^2 + ||J||^2) +
C_3 \eta^{-\frac{N}{2}} ||\rho||\right).
\end{multline*}
The constants $C_1$, $C_2$ and $C_3$ depend only on $F$ and the dimension $n$.
Since $2\delta-2 = \gamma -2 < 0$ and $N-1>1$, we first choose $\eta$ sufficiently small so that
\[
-\frac12  \eta \, F+ C_1 \eta^{N-1} r^{2\delta-2} < 0,
\]
and then choose $||\sigma||$, $||J||$ and $||\rho||$ sufficiently small (depending on $C_1$, $F$, $n$ and $\eta$), so that
\[
-\frac12 \eta \, F + r^{2\delta-2} \left( C_2 \eta^{-N-1} ( ||\sigma||^2 + || J ||^2) + C_3 \eta^{-\frac{N}{2}} ||\rho||\right) < 0
\]
as well.  This proves that $\phi_+$ is a global supersolution.
\end{proof}

\begin{theorem}{\bf (Near-CMC Global Super-Solution )}\label{nearcmcgss}
Let $(M,g)$ be AE with $\calY([g]) > 0$, and fix any $\rho \in L^{p}_{2\delta-2}$ with $\rho \geq 0$, $J\in  L^{p}_{\delta-2}$ and
$\sigma \in L^{\infty}_{\delta-1}$. Suppose that $\tau \in W^{1,p}_{\delta-1}$ satisfies $\tau - Br^{2\delta-2} \|d \tau\|_{L^p_{\delta-2}} > 0$
for some constant $B$ depending only on the dimension $n$ and the constant appearing in \eqref{DWest}.
Then there exists a global supersolution for \eqref{ce2}.
\end{theorem}
\begin{proof}
We first claim that we can choose $u \in W^{2,p}_\gamma$ such that
\[
(\Delta - c_n R) (1+u)  - b_n \tau^2 (1+u)^{N-1} = 0.
\]
This prescribed scalar curvature problem has a solution by \cite[Sec VII]{const-asympt-flat} and $1+u > 0$ by
the maximum principle. Next define $v \in W^{2,p}_{2\delta}$ by
\[
\nabla ((1+u)^2\nabla v) - b_n \tau^2 (1+v) =  - c_n (\rho + |\sigma|^2);
\]
its existence and uniqueness is guaranteed by Proposition \ref{veclapfred}, and as before, $1+v > 0$.
(Strictly speaking, we have  only stated that result for $P = \Delta - c_n R$, but the proof applies
equally well to this operator.) Now set $\phi_+ = \eta uv$, where the constant $\eta$ is chosen below. We calculate that
\[
u( \Delta - c_n R) \phi_+ = \eta ( - c_n \rho - c_n |\sigma|^2 +  b_n \tau^2 v + b_n \tau^2 u^{N}v),
\]
so for any $0< \phi < \phi_+$,
\begin{multline*}
u \, \calN(\phi_+, W(\phi)) = \eta (  - c_n\rho -c_n|\sigma|^2 +  b_n \tau^2 v + b_n \tau^2 u^{N}v)  \\
- (\eta v)^{N-1} b_n \tau^2 u^N + c_n |\sigma + \calD W|^2 u^{-N}(\eta v)^{-N-1} +
c_n (\eta v)^{-\frac{N}{2}}\rho u^{-\frac{N-2}{2}} \\
\leq b_n \tau^2 ( \eta u^N v + \eta v - (\eta v)^{N-1}u^N)  -c_n \eta (\rho + |\sigma|^2) \\
+ 2 c_n(|\sigma|^2  + |\calD W|^2) (\eta v)^{-N-1} u^{-N} + c_n \eta^{-\frac{N}{2}}\rho u^{-\frac{N-2}{2}}v^{-\frac{N}{2}}.
\end{multline*}
By \eqref{DWestim} and the inequality $\phi< \phi_+$, we have
\begin{multline*}
|\calD W|^2 \leq C r^{2\delta-2} ( (\sup \phi)^{2N}||d \tau||_{L^p_{\delta-1}}
+ ||J||_{L^p_{\delta-2}})^2  \\
\leq C' r^{2\delta-2} ( (\eta u v)^{2N}||d \tau||_{L^p_{\delta-1}}^2 +  ||J||^2_{L^p_{\delta-2}}),
\end{multline*}
and this leads to the estimate
\begin{eqnarray*}
u \, \calN(\phi_+, W(\phi))
&\leq&  c_n |\sigma|^2 \left(-\eta + 2 \eta^{-N-1}v^{-N-1}u^{-N}\right)\\
&+& c_n \rho \left( -\eta + \eta ^{-\frac{N}{2}}v ^{-\frac{N}{2} }u^{\frac{-N-2}{2}}\right)\\
&+& \eta^{N-1}\left(-\frac{b_n}{3}  \tau^2 +C_1 ||d \tau||_{L^p_{\delta-1}}^2 r^{2\delta -2} \right) v^{N-1} u^N \\
&+& \left(-\frac{b_n}{3} \eta ^{N-1} \tau^2 v^{N-1} u^N  +\eta v b_n \tau^2 +\eta v b_n \tau^2 u^N\right)\\
&+& \left(-\frac{b_n}{3}  \eta^{N-1} \tau^2 v^{N-1} u^N + C_2 \eta^{-N-1} v^{-N-1} u^{-N} r ^{2\delta -2}  ||J||_{L^p_{\delta-2}}\right).
\end{eqnarray*}
All five terms here can be made negative. Indeed, the constant $C_1$ in the third term depends on the constant in \eqref{DWest},
so we can apply the near-CMC assumption hypothesis here; the other terms are negative so long as  $\eta$ is sufficiently large.
\end{proof}

We now turn to the construction of a global subsolution. This turns out to be the same for both the far-CMC and near-CMC cases.
\begin{theorem}{\bf (Global Subsolution) }\label{gss}
Let $(M, g, \sigma, \tau, \rho, J)$ be a set of conformal data satisfying the hypotheses of either Theorem \ref{FarC} or Theorem
\ref{NearC}. Let $\psi \in W^{2,p}_\delta$ be chosen so that $1 + \psi > 0$ and $\tilde g = (1+\psi)^{N-2} g$ has scalar
curvature $R_{\tilde g} = -\frac{n-1}{n} \tau^2$ (see \cite[Sec VII]{const-asympt-flat} ). Then $\alpha (1+\psi)$ is a global subsolution
for any $0 < \alpha \leq 1$.
\end{theorem}
\begin{proof}
With this definition of $\psi$, let $\phi_- = \alpha (1+\psi)$. Then
\begin{multline*}
\calN(\phi_- ,W(\phi) ) = b_n \tau^2 (1+\psi)^{N-1} (\alpha - \alpha^{N-1}) \\
+ |\sigma + \calD W(\phi)|^2 (\alpha(1+ \psi))^{-N-1} +
c_n\rho (\alpha(1+ \psi))^{-\frac{N}{2}} \geq 0,
\end{multline*}
as required. Note that this does not even require that $\phi \geq \phi_-$.
\end{proof}

Since $\phi_- \to \alpha$ and $\phi_+ \to \eta$ at infinity, and both are strictly positive, we can choose $\alpha$
sufficiently small so that $\phi_- < \phi_+$ everywhere.

%%%%%%%%%%%%%%%%%%%%%%%%%%%%%%%%%%%%%%%%%%%%%%%%%%%%%%%%%%%%%%%%%%%%%%%%%%%%%%
\section{Fixed point Theorem and proof of the main results}
\label{sec:FixedPoint}
%%%%%%%%%%%%%%%%%%%%%%%%%%%%%%%%%%%%%%%%%%%%%%%%%%%%%%%%%%%%%%%%%%%%%%%%%%%%%%
Just as for the analogous far-CMC results on closed manifolds \cite{closed-fcmc} and \cite{vac-fcmc}, once the existence of
global sub- and supersolutions has been established, then the existence of a solution $(\phi, W)$ to \eqref{ce2}
is obtained using the Schauder fixed point theorem. Since this proof is quite similar to the one for closed
manifolds, we only sketch it here.
\begin{theorem}
For AE conformal data sets satisfying the hypotheses of Theorem \ref{FarC} or Theorem \ref{NearC},
there exists a solution $(\phi, W)$ to \eqref{ce2}, with $\phi_-\leq \phi \leq \phi_+$.
Moreover, on each end $E_j$ of $M$, $\phi- A_j \in W^{2,p}_\gamma$ for some constant $A_j$ on each end
$E_j$ of $M$.
\end{theorem}
\begin{proof}
Let $\calC^0_+$ denote the set of strictly positive bounded functions on $M$. If $\phi \in \calC^0_+$, then
by Proposition \ref{veclapfred}, the vector field $W(\phi) \in W^{2,p}_{\delta}$ is well-defined.  Next, let $T(W)$
be the solution $\phi$ to $\calN(\phi, W) = 0$ for any $W \in W^{2,p}_\delta$. This map is also well-defined.
We claim that any $\phi = T(W)$ can be decomposed as $\phi = A_j + \hat{\phi}$ on each end of $M$,
where $\hat{\phi} \in W^{2,p}_{\delta}$. (Thus if we let $A$ be a smooth function which equals $A_j$
on each end, then $\phi = A + \hat{\phi}$.) Granting this for the moment, let $S$ denote the compact
inclusion $\RR \oplus W^{2,p}_\delta \hookrightarrow \calC^0$.
A solution $(\phi, W)$ to \eqref{ce2} corresponds to a fixed point of the mapping $Q = S \circ T \circ W$. The continuity
of $W$ and $S$ are obvious, while the continuity of $T$ follows from the implicit function theorem. Up to
the claim about the decomposition of $\phi$ stated above, this proves that $Q$ is a continuous compact mapping.

Define the bounded convex set $\calS: = \{\phi \in \calC^0_+ : \phi_-\leq \phi \leq \phi_+\}$. By construction,
$Q$ maps $\calS$ to itself, and hence $Q(\calS)$ is relatively compact. Denote by $\calH$ its closed convex hull.
Thus $\calH \subset \calS$, and $Q: \calH \to \calH$. By the Schauder fixed point theorem, $\calH$ contains a
fixed point $\phi$ of $Q$. Standard estimates imply that $\phi$ and $W(\phi)$ both have the desired regularity.

The proof is finished once we prove that $T(W) = \phi = A + \hat{\phi}$, as claimed earlier. For this we rewrite the
Lichnerowicz equation as $\Delta u = f \in L^p_{\delta-2}$, where $f$ incorporates a number of terms involving $u$.
Since $\Delta: W^{2,p}_\delta \to L^p_{\delta-2}$ is an isomorphism, there exists a function $\hat{u} \in W^{2,p}_\delta$
such that $\Delta \hat{u}$ is equal to this same function $f$, hence $w = u - \hat{u}$ is a bounded harmonic
function on $M$. It is well known that on a manifold with asymptotically Euclidean ends, any such function tends
to a constant on each end. Furthermore, given any constants $A_j$, there is a bounded harmonic function
which tends to $A_j$ on $E_j$.  Indeed, define $A = \sum \chi_j A_j$ as above, where each $\chi_j$ is a cutoff
function which equals $1$ on the end $E_j$ and vanishes elsewhere. Then $\Delta A \in L^p_{\delta-2}$, so there
exists a function $\hat{A} \in W^{2,p}_\delta$ such that $\Delta\hat{A} = \Delta A$, so $A - \hat{A}$ is the bounded
harmonic function in question.
\end{proof}

\subsection*{Acknowledgments}
All  authors  are partially supported by the NSF FRG grant DMS-1263431. As well,  RM is partially supported by NSF DMS-1105050, and JI by PHY-1306441. We thank MSRI for support during the period in which this research was carried out.
The authors thank Michael Holst and David Maxwell for helpful conversations.

\end{document}